\documentclass[a4paper,english,cleveref, autoref]{lipics-v2019}

\title{Neighborhood inclusions for minimal dominating sets enumeration: linear and polynomial delay algorithms in $P_7$-free and $P_8$-free chordal graphs} 

\titlerunning{Neighborhood inclusions for minimal dominating sets enumeration}

\author{Oscar Defrain}{Universit\'e Clermont Auvergne, France. \and \url{oscar.defrain@uca.fr}}{}{}{}

\author{Lhouari Nourine}{Universit\'e Clermont Auvergne, France. \and \url{lhouari.nourine@uca.fr}}{}{}{}

\authorrunning{O. Defrain and L. Nourine}

\Copyright{Oscar Defrain and Lhouari Nourine}



\category{}

\relatedversion{}

\supplement{}

\acknowledgements{The authors are supported by the ANR project GraphEn ANR-15-CE40-0009.}

\nolinenumbers 

\hideLIPIcs  

\EventEditors{John Q. Open and Joan R. Access}
\EventNoEds{2}
\EventLongTitle{42nd Conference on Very Important Topics (CVIT 2016)}
\EventShortTitle{CVIT 2016}
\EventAcronym{CVIT}
\EventYear{2016}
\EventDate{December 24--27, 2016}
\EventLocation{Little Whinging, United Kingdom}
\EventLogo{}
\SeriesVolume{42}
\ArticleNo{23}

\usepackage[textwidth=2.5cm, color=yellow]{todonotes}
\usepackage[linesnumbered, vlined, ruled]{algorithm2e}
\usepackage[noadjust]{cite}
\usepackage{float}
\usepackage{needspace}

\def\D{\mathcal{D}} 

\def\DIR{DIR}

\newcommand{\intv}[2]{[#2]}

\EventEditors{Pinyan Lu and Guochuan Zhang}
\EventNoEds{2}
\EventLongTitle{30th International Symposium on Algorithms and Computation (ISAAC 2019)}
\EventShortTitle{ISAAC 2019}
\EventAcronym{ISAAC}
\EventYear{2019}
\EventDate{December 8--11, 2019}
\EventLocation{Shanghai University of Finance and Economics, Shanghai, China}
\EventLogo{}
\SeriesVolume{149}
\ArticleNo{66}

\begin{document}

\maketitle

\begin{abstract}
	In [M.~M.~Kant\'e, V.~Limouzy, A.~Mary, and L.~Nourine.~On the enumeration of minimal dominating sets and related notions.~SIAM Journal on Discrete Mathematics, 28(4):1916–1929, 2014.] the authors give an $O(n+m)$ delay algorithm based on neighborhood inclusions for the enumeration of minimal dominating sets in split and $P_6$-free chordal graphs.
	In this paper, we investigate generalizations of this technique to $P_k$-free chordal graphs for larger integers $k$.
	In particular, we give $O(n+m)$ and $O(n^3\cdot m)$ delays algorithms in the classes of $P_7$-free and $P_8$-free chordal graphs.
	As for $P_k$-free chordal graphs for $k\geq 9$, we give evidence that such a technique is inefficient as a key step of the algorithm, namely the irredundant extension problem, becomes {\sf NP}-complete.
\end{abstract}

\section{Introduction}

We consider the problem of enumerating all inclusion-wise minimal dominating sets of a given graph, denoted by \textsc{Dom-Enum}.
A {\em dominating set} in a graph $G$ is a set of vertices $D$ such that every vertex of $G$ is either in $D$ or is adjacent to some vertex of $D$. 
It is said to be {\em minimal} if it does not contain any dominating set as a proper subset.
To this date, it is open whether \textsc{Dom-Enum} admits an output-polynomial time algorithm.
An enumeration algorithm is said to be running in {\em output-polynomial} time if its running time is bounded by a polynomial in the combined size of the input and the output.
It is said to be running in {\em incremental-polynomial} time if the running times between two consecutive outputs and after the last output are bounded by a polynomial in the combined size of the input and already output solutions.
If the running times between two consecutive outputs and after the last output are bounded by a polynomial in the size of the input alone, then the algorithm is said to be running with {\em polynomial delay}; see~\cite{johnson1988generating,creignou2019complexity}.
Recently, it has been proved in~\cite{kante2014enumeration} that \textsc{Dom-Enum} is equivalent to the problem of enumerating all inclusion-wise minimal transversals of a hypergraph, denoted by \textsc{Trans-Enum}.
The best known algorithm for this problem is due to Fredman and Khachiyan~\cite{fredman1996complexity} and runs in incremental quasi-polynomial time.
Nevertheless, several classes of graphs were shown to admit output-polynomial time algorithms.
For example, it has been shown that there exist output-polynomial time algorithms for $\log(n)$-degenerate graphs~\cite{eiter2003new}, triangle-free graphs~\cite{bonamy2019enumerating}, and recently for $K_t$-free for any fixed $t\in \mathbb{N}$, diamond-free and paw-free graphs~\cite{bonamy2019enumeratingkt}.
Incremental-polynomial time algorithms are known for chordal bipartite graphs~\cite{golovach2016enumerating} and graphs of bounded conformality~\cite{boros2004generating}.
Polynomial-delay algorithms are known for degenerate graphs~\cite{eiter2003new}, line graphs~\cite{kante2015polynomial}, and chordal graphs~\cite{kante2015chordal}. 
Linear-delay algorithms are known for permutation and interval graphs~\cite{kante2012neighbourhood}, graphs with bounded clique width~\cite{courcelle2009linear}, split and $P_6$-free chordal graphs~\cite{kante2014enumeration}.

In this paper, we investigate the enumeration of minimal dominating sets from their intersection with redundant vertices, i.e., vertices that have an inclusion-wise non-minimal neighborhood in the graph. 
This technique was first introduced in~\cite{kante2014enumeration} for the enumeration of minimal dominating sets in split and $P_6$-free chordal graphs.
We investigate generalizations of this technique to $P_k$-free chordal graphs for larger integers $k$.
In particular, we give $O(n+m)$ and $O(n^3\cdot m)$ delays algorithms in the classes of $P_7$-free and $P_8$-free chordal graphs, where $n$ and $m$ respectively denote the number of vertices and edges in the graph.
Our algorithms rely on two main properties.
The first one is that the intersections of minimal dominating sets with redundant vertices form an independence system and an accessible set system in $P_7$-free and $P_8$-free chordal graphs.
The second is that the connected components obtained after removing redundant vertices in $P_7$-free and $P_8$-free chordal graphs are respectively $P_3$-free and $P_4$-free chordal.
As for $P_k$-free chordal graphs for $k\geq 9$, we give evidence that such a technique is inefficient as a key step of the algorithm, namely the irredundant extension problem, becomes {\sf NP}-complete.

The rest of the paper is organized as follows. 
In Section~\ref{sec:preliminaries} we introduce definitions and preliminary notions.
In Section~\ref{sec:algorithm} we describe the general algorithm that we consider throughout the paper and that can be decomposed into two distinct parts: redundant parts enumeration, and irredundant extensions enumeration.
In Section~\ref{sec:properties} we prove properties on chordal graphs that depend on the size of a longest induced path in the graph.
Section~\ref{sec:redundant} is devoted to the complexity analysis of the first part of the algorithm, while Section~\ref{sec:irredundant} consider the second.
We conclude in Section~\ref{sec:conclusion} by discussing the outlooks of such a technique.

\section{Preliminaries}\label{sec:preliminaries}

In this paper, all graphs are considered finite, undirected, simple, and loopless. 
For a graph $G=(V(G),E(G))$, $V(G)$ is its set of vertices and $E(G)\subseteq \{\{x,y\} \mid x,y\in V(G),\ x\neq y\}$ is its set of edges.
Edges may be denoted by $xy$ (or $yx$) instead of $\{x,y\}$.
Two vertices $x,y$ of $G$ are called {\em adjacent} if $xy\in E(G)$.
A {\em clique} in a graph $G$ is a set of pairwise adjacent vertices.
An {\em independent set} in a graph $G$ is a set of pairwise non-adjacent vertices.
The subgraph of $G$ {\em induced} by $X\subseteq V(G)$, denoted by $G[X]$, is the graph $(X,E\cap \{\{x,y\} \mid x,y\in X,\ x\neq y\})$; $G-X$ is the graph $G[V(G)\setminus X]$.
An {\em induced path} (resp.~{\em induced cycle}) in $G$ is a path (resp.~cycle) that is an induced subgraph of~$G$.
We denote by $P_k$ an induced path on $k$ vertices.
We call {\em hole} (or {\em chordless cycle}) an induced cycle of size at least four.
A~graph $G$ is {\em split} if its vertex set can be partitioned into a clique and an independent set.
It~is {\em chordal} if it has no chordless cycle.
It~is called $P_k$-free if it has no induced path on $k$ vertices.

Let $G$ be a graph and $x\in V(G)$ be a vertex of $G$.
The {\em neighborhood} of $x$ is the set $N(x)=\{y\in V(G) \mid xy\in E(G)\}$.
The {\em closed neighborhood} of $x$ is the set $N[x]= N(x)\cup\{x\}$.
For a subset $X\subseteq V(G)$ we define $N[X]=\bigcup_{x\in X} N[x]$ and $N(X)=N[X]\setminus X$.
In case of ambiguity or when several graphs are considered, we shall note $N_G[x]$ the neighborhood of $x$ in $G$.
The {\em degree} of $x$ is defined by $deg(x)=|N(x)|$.
We say that $x$ is {\em complete} to $X$ if $X\subseteq N(x)$, and that it is {\em partially adjacent} to $X$ if it is adjacent to an element of $X$ but not complete to $X$.
Let $D,X\subseteq V(G)$ be two subsets of vertices of $G$.
We say that $D$ {\em dominates} $X$ if $X\subseteq N[D]$.
It is inclusion-wise {\em minimal} if $X\not\subseteq N[D\setminus \{x\}]$ for any $x\in D$.
We say that $D$ dominates $x$ if it dominates $\{x\}$.
A (minimal) {\em dominating set} of $G$ is a (minimal) dominating set of $V(G)$.
The set of all minimal dominating sets of $G$ is denoted by $\D(G)$, and the problem of enumerating $\D(G)$ given $G$ by \textsc{Dom-Enum}.
Let $x$ be a vertex of~$D$.
A~{\em private neighbor} of $x$ w.r.t.~$D$ in $G$ is a vertex $u$ of $G$ that is only adjacent to $x$ in~$D$, that is, such that $N[u]\cap D=\{x\}$. 
Note that $x$ can be its own private neighbor (in that case we say that $x$ is {\em self-private}).
The set of all private neighbors of $x$ w.r.t.~$D$ is denoted by $Priv(D,x)$.
It is well known that a subset $D\subseteq V(G)$ is a minimal dominating set of $G$ if and only if it dominates $G$, and for every $x\in D$, $Priv(D,x)\neq \emptyset$.

Let $x$ be a vertex of $G$.
We say that $x$ is {\em irredundant} if it is minimal with respect to neighborhood inclusion.
In case of equality between minimal neighborhoods, exactly one vertex is considered as irredundant.
We say that $x$ is {\em redundant} if it is not irredundant.
Then to every redundant vertex $y$ corresponds at least one irredundant vertex $x$ such that $N[x]\subseteq N[y]$, and no vertex $y$ is such that $N[y]\subset N[x]$ whenever $x$ is irredundant.
The set of irredundant vertices of $G$ is denoted by $IR(G)$, and the set of redundant vertices by $RN(G)$.
We call {\em irredundant component} a connected component of $G[IR(G)]$.
For a subset $D$ of vertices of $G$ we note $D_{RN}=D\cap RN(G)$ its intersection with redundant vertices, and $D_{IR}=D\cap IR(G)$ its intersection with irredundant vertices.
Then $D_{RN}$ and $D_{IR}$ form a bipartition of $D$.
For a subset $D$ and a vertex $x\in D$, we call irredundant private neighbors of $x$ w.r.t.~$D$ the elements of the set $Priv_{IR}(D,x)=Priv(D,x)\cap IR(G)$.
In the remaining of the paper we shall note $\D_{RN}(G)=\{D_{RN} \mid D\in \D(G)\}$ and refer to this set as the {\em redundant parts} of minimal dominating sets of $G$.
We call {\em irredundant extension} of $A\in \D_{RN}(G)$ a set $I\subseteq IR(G)$ such that $A\cup I\in \D(G)$, and note $\DIR(A)$ the set of all such sets.
Observe that $|\D_{RN}(G)|\leq |\D(G)|$ and that this inequality might be sharp (take a star graph), or strict (take a path on six vertices).
We end the preliminaries stating general properties that will be used throughout the paper.

\begin{proposition}\label{prop:ACS-emptyset}
	Let $G$ be a graph.
	Then $IR(G)$ dominates $G$, hence $\emptyset\in \D_{RN}(G)$.
\end{proposition}

\begin{proof}
	Take any vertex $x$ of $G$.
	Either it is irredundant, or not.
	If it is then it is dominated by $IR(G)$.
	If not then by definition there exists $y\in IR(G)$ such that $N[y]\subseteq N[x]$, and it is dominated by $IR(G)$.
	Consequently, $IR(G)$ dominates $G$ and thus there exists $D\subseteq IR(G)$ such that $D\in \D(G)$ and $D_{RN}=\emptyset$.
	Hence $\emptyset\in \D_{RN}(G)$.
\end{proof}

\begin{proposition}\label{prop:IR-private}
	Let $G$ be a graph and $D\subseteq V(G)$.
	Then $D$ is a minimal dominating set of $G$ if and only if it dominates $IR(G)$ and $Priv_{IR}(D,x)\neq\emptyset$ for every $x\in D$. 
\end{proposition}

\begin{proof}
	We prove the first implication.
	Let $D\in \D(G)$.
	Clearly $D$ dominates $IR(G)$. 
	Let us assume for contradiction that $Priv_{IR}(D,x)=\emptyset$ for some $x\in D$.
	We first exclude the case where $x$ is self-private.
	If $x$ is self-private then it is redundant and it has a neighbor $y\in IR(G)$ such that $N[y]\subseteq N[x]$.
	Since by hypothesis $Priv_{IR}(D,x)=\emptyset$, $y$ is dominated by some $z\in D$, $x\neq z$.
	However, since $N[y]\subseteq N[x]$ then $zx\in E(G)$ and $x$ is not self-private, a contradiction.
	Consequently $x$ has a neighbor $u\in D$, and a private neighbor $v$ in $RN(G)$.
	Let $w\in IR(G)$ such that $N[w]\subseteq N[v]$.
	Such a vertex exists since $v$ is redundant.
	Two cases arise depending on whether $w=x$ or $w\neq x$.
	In the first case we conclude that $uv\in E(G)$, hence that $v$ is not a private neighbor of $x$, a contradiction.
	In the other case, observe that since $w$ is irredundant it cannot be a private neighbor of $x$ (if ever it was adjacent to $x$).
	Hence it must be dominated by some $z\in D$, $z\neq x$ (possibly $z=w)$.
	Since $N[w]\subseteq N[v]$, $z$ is adjacent to $v$, hence $v$ is not a private neighbor of $x$, a contradiction.

	As for the other implication, observe that if an irredundant neighborhood $N[x]$, $x\in IR(G)$ is intersected by some set $D\subseteq V(G)$, then every neighborhood $N[y]$ such that $N[x]\subseteq N[y]$ is also intersected by $D$.
	Now if $D$ dominates $IR(G)$, then it intersects every irredundant neighborhood.
	As for every $y\in RN(G)$ there exists $x\in IR(G)$ such that $N[x]\subseteq N[y]$ we conclude that $D$ dominates $G$ whenever it dominates $IR(G)$.
	Minimality follows from the inclusion $Priv_{IR}(D,x)\subseteq Priv(D,x)$, recalling that a dominating set $D$ is minimal if and only if $Priv(D,x)\neq\emptyset$ for every $x\in D$.
\end{proof}

A corollary of Proposition~\ref{prop:IR-private} is the following, observing for $A\subseteq RN(G)$ and $I\subseteq IR(G)$ that if $I$ dominates $IR(G)\setminus N(A)$ but not $Priv_{IR}(A,a)$ for any $a\in A$, then $I$ can be arbitrarily reduced into a minimal such set.

\begin{corollary}\label{cor:IR-private}
	Let $G$ be a graph and $A\subseteq RN(G)$.
	Then $A\in \D_{RN}(G)$ if and only if every $a\in A$ has an irredundant private neighbor, and there exists $I\subseteq IR(G)$ such that $I$ dominates $IR(G)\setminus N(A)$ but not $Priv_{IR}(A,x)\neq\emptyset$ for any $a\in A$.
	Furthermore, $I\in DIR(A)$ whenever it is minimal with this property.
\end{corollary}

\section{The algorithm}\label{sec:algorithm}

We describe a general algorithm enumerating the minimal dominating sets of a graph from their intersection with redundant vertices.
See Algorithm~\ref{algo:main}.
The first step is the enumeration of such intersections, Line~\ref{line:main-forallRN}. 
The second step is the enumeration of their irredundant extensions, Line~\ref{line:main-forallIR}.
The correctness of the algorithm follows from the bipartition induced by $RN(G)$ and $IR(G)$ in $G$.

The next sections are devoted to the complexity analysis of these two steps in the restricted case of $P_7$-free and $P_8$-free chordal graphs.

\begin{algorithm}
	\SetAlgoLined
	
	\SetKwProg{myproc}{Procedure}{}{}
	\myproc{{\em \texttt{DOM}($G$)}}{
		\For{{\bf all} $A\subseteq RN(G)$ {\bf such that} $A\in \D_{RN}(G)$\label{line:main-forallRN}}
		{	
			\For{{\bf all} $I\subseteq IR(G)$ {\bf such that} $I\in \DIR(A)$\label{line:main-forallIR}}
			{	
				{\bf output} $A\cup I$\;\label{line:main-output}
			}	
		}
	}
	\caption{An algorithm enumerating the minimal dominating sets of a graph $G$ from their intersection with the set $RN(G)$ of redundant vertices of $G$.}\label{algo:main}
\end{algorithm}

\section{Properties on $P_k$-free chordal graphs}\label{sec:properties}

We give structural properties on redundant vertices and irredundant components of $G$ whenever $G$ is chordal, and depending on the size of a longest induced path in $G$.

\begin{proposition}\label{prop:IR-path}
	Let $G$ be a graph and $u,v$ be two adjacent irredundant vertices of $G$. 
	Then there exist $u'\in N[u]\setminus N[v]$, $u''\in N[u']\setminus N[u]$, $v'\in N[v]\setminus N[u]$ and $v''\in N[v']\setminus N[v]$.
	In particular if $G$ is chordal, then $u''u'uvv'v''$ induces a $P_6$.
\end{proposition}

\begin{proof}
	Let us assume for contradiction that no such $u'$ exists.
	Then either $N[u]\subset N[v]$, or $N[u]=N[v]$.
	In the first case $v$ is redundant, a contradiction.
	In the other case only one of $u$ and $v$ should be irredundant by definition, a contradiction.
	Hence $u'$ exists. 
	By symmetry, $v'$ exists.
	Let us now assume for contradiction that no such $u''$ exists.
	Then either $N[u']\subset N[u]$, or $N[u']=N[u]$.
	In the first case $u$ is redundant, a contradiction.
	In the other case $vu'\in E(G)$, a contradiction.
	Hence $u''$ exists.
	By symmetry, $v''$ exists.
	Now if $G$ is chordal, $u''u'uvv'v'$ induces a $P_6$.
\end{proof}

An {\em accessible set system} is a family of sets in which every non-empty set $X$ contains an element $x$ such that $X\setminus\{x\}$ belongs to the family.
If $x$ is of largest index in $X$ such that $X\setminus\{x\}$ belongs to the family, then it is called {\em maximal generator} of $X$.
An \emph{independence system} is a family of sets such that for every non-empty set $X$ of the family, and every element $x\in X$, $X\setminus\{x\}$ belongs to the family.
In particular, every independence system is an accessible set system.
Note that the maximal generator of $X$ in that case is always the vertex of maximal index in $X$.
Accessible set systems and independence systems play an important role in the design of efficient enumeration algorithms~\cite{arimura2009polynomial,kante2014enumeration}. 
The next theorem suggests that the enumeration of $\D_{RN}(G)$ is tractable in $P_7$-free and $P_8$-free chordal graphs. 

\begin{figure}
	\center
	\caption{The situation of Proposition~\ref{prop:ASC-IS}, case one. Circles denote private neighborhoods.}\label{fig:P8ASC}
	\includegraphics{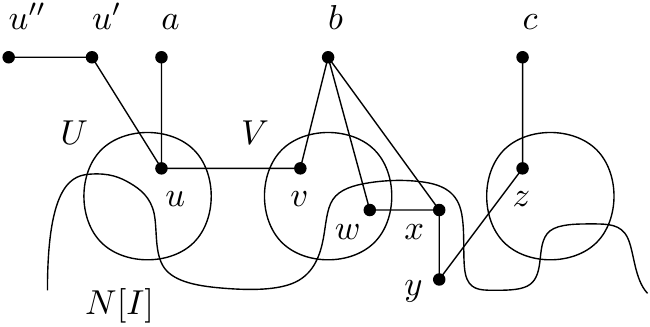}
\end{figure}

\begin{proposition}\label{prop:ASC-IS}
	Let $G$ be a chordal graph. 
	Then $\D_{RN}(G)$ is an independence system whenever $G$ is $P_7$-free, and it is an accessible set system whenever $G$ is $P_8$-free.
\end{proposition}

\begin{proof}
	Let $G$ be a chordal graph.
	We first assume that $\D_{RN}(G)$ is not an independence system to exhibit a $P_7$, and then assume that $\D_{RN}(G)$ is not an accessible system to exhibit a $P_8$.
	So suppose that $\D_{RN}(G)$ is not an independence system and let $A\in\D_{RN}(G)$ and $a\in A$ such that $A\setminus \{a\}\not\in \D_{RN}(G)$.
	By Proposition~\ref{prop:ACS-emptyset}, $|A|\geq 2$.
	Let $I\in DIR(A)$.
	Clearly $Priv_{IR}(A,a)\not\subseteq I$.
	Let $A'=A\setminus \{a\}$ and $I'=I\cup Priv_{IR}(A,a)$.
	Then $I'$ dominates $IR(G)\setminus N(A')$.
	By Corollary~\ref{cor:IR-private} there must be some $b\in A'$ such that $I'$ dominates $Priv_{IR}(A',b)$, hence $Priv_{IR}(A,b)$ as $Priv_{IR}(A,b)\subseteq Priv_{IR}(A',b)$.
	Let $b$ be one such vertex. 
	We put $U=Priv_{IR}(A,a)\setminus N[I]$ and $V=Priv_{IR}(A,b)\setminus N[I]$.
	Then neither of $U$ nor $V$ is empty, $U\cap V=\emptyset$, and $U$ dominates~$V$.
	Let $u\in U$ and $v\in V$ be such that $uv\in E(G)$ (such $u$ and $v$ exist since $U$ dominates $V$).
	Since $u$ and $v$ are private neighbors of $a$ and $b$, $av,bu\not\in E(G)$.
	Since $G$ is chordal, $ab\not\in E(G)$.
	Then $auvb$ induces a $P_4$.
	By Proposition~\ref{prop:IR-path} since $u,v$ are irredundant, there exists $u''$ and $u'$ such that $u''u'uvb$ induces a $P_5$.
	Consider an irredundant vertex $w$ such that $N[w]\subseteq N[b]$.
	Such a vertex exists since~$b$ is redundant.
	Two cases arise depending on whether $w\in Priv_{IR}(A,b)$ or $w\not\in Priv_{IR}(A,b)$.
	
	Let us consider the case $w\in Priv_{IR}(A,b)$.
	It is illustrated in Figure~\ref{fig:P8ASC}.
	Since $U$ dominates $V$ and $N[w]\subseteq N[b]$ we know that $w\not\in V$ (as otherwise $b$ is adjacent to a vertex of $U$, i.e., a private neighbor of~$a$).
	Hence $w\in Priv_{IR}(A,b)\cap N[I]$. 
	Note that $w\not\in I$ as $N[w]\subseteq N[b]$ ($w$ cannot be part of an irredundant extension if it has no private neighbors).
	Accordingly, consider $x\in I$ such that $wx\in E(G)$.
	Since $N[w]\subseteq N[b]$, $xb\in E(G)$.
	Since $v\not\in N[I]$, $xv\not\in E(G)$.
	Now, since $x$ belongs to $I$ it has a private neighbor $y\in N[x]\setminus N[w]$.
	As $G$ is chordal $u''u'uvbxxy$ induces a $P_7$, concluding the first part of the proposition in this case.
	Let us now assume that $\D_{RN}(G)$ is not an accessible set system, that is $A\setminus \{c\}\not\in \D_{RN}(G)$ for any $c\in A$.
	Observe that if replacing $x$ by $Priv_{IR}(A'\cup I', x)$ in $I'$ for all $x\in N(w)\cap I'$ does not dominate $Priv_{IR}(A',c)$ for any $c\in A'\setminus \{b\}$, then $w$ becomes a private neighbor of~$b$, and $A\setminus \{a\} \in \D_{RN}(G)$, a contradiction.
	Consequently there must exist $x\in I$ such that $wx\in E(G)$ and $y\in Priv_{IR}(A'\cup I', x)$, $c\in A$ and $z\in Priv_{IR}(A',c)$ such that $yz\in E(G)$.
	Also $yb,zx,cy\not\in E(G)$ as $y$ and $z$ are private neighbors of $x$ and~$c$.
	Since $G$ is chordal, $u''u'uvbxyzc$ induces $P_9$, concluding the second part of the proposition in this case.

	Let us now consider the other case $w\not\in Priv_{IR}(A,b)$.
	Then there must exist $c\in A\setminus \{b\}$ such that $wc\in E(G)$.
	Since $N[w]\subseteq N[b]$, we have $bc\in E(G)$.
	Consequently $a\neq c$.
	Furthermore since $v$ is a private neighbor of $b$, $cv\not\in E(G)$.
	Since $c\in A$ it has a private neighbor $z$, and $bz\not\in E(G)$.
	As $G$ is chordal $u''u'uvbcz$ induces a $P_7$, concluding the first part of the proposition in this second case.
	Let us now assume that $\D_{RN}(G)$ is not an accessible set system.
	Then $A\setminus \{c\}\not\in \D_{RN}(G)$.
	Observe that if every private neighbor $z$ of $c$ is such that $N[z]\subseteq N[c]$, then replacing $c$ by every such private neighbors in $A\cup I$ yields a minimal dominating set $D$ of $G$ such that $D_{RN}=A\setminus \{c\}$, a contradiction.
	Hence there exist $z\in N[c]\setminus N[b]$ and $z'\in N[z]\setminus N[c]$.
	As $G$ is chordal $u''u'uvbczz'$ induces a $P_8$, concluding the second part of the proposition in this case, and the proof.
\end{proof}

\begin{proposition}\label{prop:Pk-Pk-4}
	Let $G$ be a chordal graph and $C$ be an irredundant component of~$G$.
	Then the graph $G[C]$ is $P_{k-4}$-free chordal whenever $G$ is $P_k$-free, $k\geq 6$.
\end{proposition}

\begin{proof}
	We proceed by contradiction.
	Let $G$ be a $P_k$-free graph, $k\geq 6$ and $C$ be an irredundant component of $G$.
	Suppose that $G[C]$ is not $P_{k-4}$-free, and let $P_{uv}$ be an induced path of length at least $k-4$ in $G[C]$ with endpoints $u$ and~$v$.
	Let $u^*$ and $v^*$ be the neighbors of $u$ and $v$ in $P_{uv}$ (possibly $u^*=v$ and $v^*=u$, or $u^*=v^*$).
	By Proposition~\ref{prop:IR-path} since $u,u^*$ and $v,v^*$ are irredundant and adjacent, there exist $u'',u',v',v''$ such that $u''u'P_{uv}v'v''$ induces a path of length at least $k$ in $G$, a contradiction.
\end{proof}

\begin{proposition}\label{prop:Na-connected}
	Let $G$ be a chordal graph, $a\in RN(G)$, $C$ be an irredundant component of~$G$, and $u,v$ be two vertices in $C\cap N(a)$.
	Then $N(a)$ contains every induced path from $u$ to~$v$.
	In particular $G[N(a)\cap C]$ is connected.
\end{proposition}

\begin{proof}
	Clearly the proposition holds if $uv\in E(G)$.
	Let $u,v$ be two non-adjacent vertices in $C\cap N(a)$. 
	Let $P_{uv}$ be an induced path from $u$ to $v$ in $G[C]$.
	One such path exists since $G[C]$ is connected.
	Let us assume for contradiction that there exists $x\in P_{uv}$ such that $x\not\in N(a)$.
	Consider $u^*$ and $v^*$ to be the first elements of $P_{uv}$ respectively in the way from $x$ to $u$, and from $x$ to $v$, such that $u^*,v^*\in N(a)$ (possibly $u^*=u$ and $v^*=v$).
	Consider the path $P_{u^*v^*}$ obtained from $P_{uv}$ and shortened at endpoints $u^*$ and $v^*$.
	Then $P_{u^*v^*}$ is an induced path with only its endpoints adjacent to $a$, inducing a hole in $G$, a contradiction.
\end{proof}

\begin{proposition}\label{prop:Na-complete}
	Let $G$ be a chordal graph and $a\in RN(G)$.
	Then $a$ is partially adjacent to at most one irredundant component of $G$ (it is either disconnected or complete to all other irredundant components of $G$) whenever $G$ is $P_9$-free chordal.
\end{proposition}

\begin{proof}
	We proceed by contradiction.
	Let us assume that $G$ is $P_9$-free chordal and that there exist two irredundant components $C_1,C_2$ such that $C_1\cap N(a)\neq \emptyset$, $C_2\cap N(a)\neq \emptyset$, and $C_1,C_2\not\subseteq N(a)$.
	Let $u\in C_1\cap N(a)$, $u'\in C_1\setminus N(a)$, $v\in C_2\cap N(a)$ and $v'\in C_2\setminus N(a)$.
	Consider a shortest path $P_{u'u}$ in $G[C_1]$ from $u'$ to $u$, and one $P_{vv'}$ in $G[C_2]$ from $v$ to $v'$.
	These paths are induced.
	Let $u^*$ and $v^*$ be the neighbors of $u'$ and $v'$ in $P_{u'u}$ and $P_{vv'}$, respectively (possibly $u^*=u$ and $v^*=v$).
	By Proposition~\ref{prop:IR-path} since $u',u^*$ and $v',v^*$ are irredundant and adjacent, there exist $u''$, $u'''$, $v''$ and $v'''$ such that $u'''u''u'u^*$ and $v^*v'v''v'''$ induce paths of length four in $G$.
	Consider $x$ the last vertex in $P_{u'u}$ starting from $u$ which is adjacent to $a$, and $y$ the last vertex in $P_{vv'}$ starting from $v$ which is adjacent to $a$ (possibly $x=u^*$ and $y=v^*$ but $x\neq u'$, $y\neq v'$).
	Consider the paths $P_{u'x}$ and $P_{yv'}$ obtained from $P_{u'u}$ and $P_{vv'}$ and shortened at endpoints $x$ and $y$.
	Then $u'''u''P_{u'x}aP_{yv'}v''v'''$ induces a path of length at least nine in $G$, a contradiction.
\end{proof}

In the following, for a set $A\subseteq RN(G)$ we consider the following bipartition.
The part $B(A)$ contains the elements of $A$ having an irredundant private neighbor in some irredundant component $C$ such that $C\subseteq N(A)$.
Observe that no irredundant extension of $A$ can steal these private neighbors, as only $IR(G)\setminus N(A)$ has to be dominated by such extensions, and $C$ is disconneced from $IR(G)\setminus N(A)$.
The part $R(A)$ contains all other elements of~$A$.
We call {\em red} and {\em blue} vertices the elements of $R(A)$ and $B(A)$, respectively.
If $C_i$ is an irredundant component of $G$, then $R_i(A)$ denote the red elements of $A$ having at least one private neighbor in $C_i$.
Recall that by Proposition~\ref{prop:Na-complete}, the elements of $A$ are partially adjacent to at most one irredundant component whenever $G$ is $P_9$-free chordal.
In particular in such class, the red elements have their private neighbors in at most one irredundant component.
The next theorem follows.

\begin{theorem}\label{thm:P9-IR-characterization}
	Let $G$ be a $P_9$-free chordal graph, $A\in \D_{RN}(G)$ and $I\subseteq IR(G)$.
	Then $I$ is an irredundant extension of $A$ if and only if for every irredundant component $C_i$ of $G$, $D_i=I\cap C_i$ is minimal such that
	\begin{itemize}
		\item $D_i$ dominates $C_i\setminus N(A)$, but
		\item $D_i$ does not dominate $Priv_{IR}(A,x)$ for any $x\in R_i(A)$.
	\end{itemize}
\end{theorem}

We immediately derive the next two corollaries, observing for the first one that a minimal set $I$ as described in Theorem~\ref{thm:P9-IR-characterization} can be greedily obtained from a non-minimal such set, and for the second that by Proposition~\ref{prop:Pk-Pk-4}, every irredundant component $C$ of $G$ is a clique whenever $G$ is $P_7$-free chordal.

\begin{corollary}\label{cor:P9-IR-characterization}
	Let $G$ be a $P_9$-free chordal graph and $A\subseteq RN(G)$.
	Then $A\in \D_{RN}(G)$ if and only if every $a\in A$ has an irredundant private neighbor, and, for every irredundant component $C_i$ of $G$ there exists $D_i\subseteq C_i$ such that
	\begin{itemize}
		\item $D_i$ dominates $C_i\setminus N(A)$, but
		\item $D_i$ does not dominate $Priv_{IR}(A,x)$ for any $x\in R_i(A)$.
	\end{itemize}
\end{corollary}

\begin{corollary}\label{cor:P7-IR-characterization}
	Let $G$ be a $P_7$-free chordal graph.
	Then $\D_{RN}(G)=\{A\subseteq RN(G) \mid$ every $x\in A$ has a private neighbor in some irredundant component $C\subseteq N(A)$, i.e., $R(A)=\emptyset\}$.
\end{corollary}

\section{Enumerating the redundant part of minimal dominating sets}\label{sec:redundant}

This section is devoted to the complexity analysis of Line~\ref{line:main-forallRN} of Algorithm~\ref{algo:main}.
More precisely, we show that enumerating the redundant part of minimal dominating sets can be done with linear and polynomial delays in $P_7$-free and $P_8$-free chordal graphs.

Recall that by Proposition~\ref{prop:ASC-IS}, the set $\D_{RN}(G)$ is an accessible set system whenever $G$ is $P_8$-free chordal.
Hence, it is sufficient to be able to decide whether (i) a given set $A\subseteq RN(G)$ belongs to $\D_{RN}(G)$, and (ii) a given vertex $c$ of $A\in\D_{RN}$ is a maximal generator of~$A$, in order to get an algorithm enumerating $\D_{RN}(G)$ without repetitions in such class.
We~call {\em irredundant extension problem} the first decision problem (denoted by \textsc{IEP}), and {\em maximal generator problem} the second (denoted by \textsc{MGP}).
The algorithm proceeds as follows.
See Algorithm~\ref{algo:RN}.
Given $A\in \D_{RN}(G)$ (starting with $A=\emptyset$ according to Proposition~\ref{prop:ACS-emptyset}) it checks for every candidate vertex $c\in RN(G)\setminus A$ whether $A\cup \{c\}$ belongs to $\D_{RN}(G)$, whether $c$ is a maximal generator of $A\cup \{c\}$, and if so, makes a recursive call on such a set.
The correctness of the algorithm follows from the fact that $\D_{RN}(G)$ being an accessible set system, every set in $\D_{RN}(G)$ is accessible by such a procedure.
In particular, every set $A$ received by the algorithm belongs to $\D_{RN}(G)$.
Repetitions are avoided by the choice of $c$.

\begin{algorithm}
	\SetAlgoLined

	\SetKwProg{myproc}{Procedure}{}{}
	\myproc{{\em \texttt{RNDom}($G$)}}{
		\texttt{RecRNDom}($G, \emptyset$)\;
	}

	\SetKwProg{myproc}{Procedure}{}{}
	\myproc{{\em \texttt{RecRNDom}($G, A$)}}{
		{\bf output} $A$\;\label{line:RN-output}

		\For{{\bf all} $c\in RN(G)\setminus A$\label{line:RN-forall}}
		{	
			\If{$A\cup \{c\}\in \D_{RN}(G)$ {\bf and} $c$ {is a maximal generator of} $A\cup \{c\}$\label{line:RN-condition}}{
				\texttt{RecRNDom}($G,A\cup\{c\}$)\;\label{line:RN-recursivecall}
			}	

		}\label{line:RN-afterloop}
	}

	\caption{An algorithm enumerating the set $\D_{RN}(G)$ of a $P_8$-free chordal graph~$G$, relying on the fact that $\D_{RN}(G)$ is an accessible set system on such class.}
    \label{algo:RN}
\end{algorithm}

\subsection{Linear delay implementation in $P_7$-free chordal graphs.}
We show that there is a linear-delay implementation of Algorithm~\ref{algo:RN} in $P_7$-free chordal graphs.
The proof is technically involved and makes use of preprocessed arrays that are maintained throughout the computation.

\begin{theorem}\label{thm:P7-RN}
	There is an $O(n+m)$ delay, $O(n^2)$ space and $O(n^2)$ preprocessing-time implementation of Algorithm~\ref{algo:RN} whenever $G$ is $P_7$-free chordal, where $n$ and $m$ respectively denote the number of vertices and edges in $G$.
\end{theorem}

\begin{proof}
	Let $C_1,\dots,C_\ell$ denote the $\ell$ irredundant components of $G$.
	For every $a\in RN(G)$, and according to Proposition~\ref{prop:Na-complete}, we note $C^a=C_i$ the unique irredundant component $C_i$ to which $a$ is partially adjacent, if it exists, and $C^a=\emptyset$ otherwise.
	Note that the computation of such components, and the identification of $C^a$ for every $a\in RN(G)$ can be done in $O(n^2)$ preprocessing time and takes $O(n^2)$ space.
	Consider $A\in \D_{RN}(G)$ as received by the algorithm.
	Let $c\in RN(G)\setminus A$.
	First observe that the condition of $c$ being a maximal generator of $A\cup \{c\}$ Line~\ref{line:RN-condition} can be implicitly verified by selecting $c$ of index greater than those in $A$.
	This can be done by computing the maximal index $\rho$ in $A$ before the loop in $O(n)$ time, and iterating on $c$ such that $c>\rho$ with no extra cost on the complexity of the loop.
	We shall show using preprocessed arrays maintained at each step of the loop that testing whether $A\cup \{c\}\in \D_{RN}(G)$ is bounded by $O(\deg(c))$.
	Note that by Corollary~\ref{cor:P7-IR-characterization}, $A\cup \{c\}$ belongs to $\D_{RN}(G)$ if and only if (i) every $a\in A$ has a private neighbor in some irredundant component $C_j\subseteq N(A\cup \{c\})$, $j\in \intv{1}{\ell}$, and (ii) there exists $C_i$, $i\in \intv{1}{\ell}$ such that $C_i\not\subseteq N(A)$ and $C_i\subseteq N(A\cup\{c\})$.
	Also, recall that by Proposition~\ref{prop:Pk-Pk-4} every component $C_1,\dots,C_\ell$ is a clique.
	Let $T_1$ be an array of size $\ell$ such that $T_1[i]=|C_i|$ for every $i\in\intv{1}{\ell}$.
	This array will be used to know the number of vertices that are yet to be dominated in every component.
	Let $T_2$ be an array of size $n$ such that $T_2[y]=i$ if $y\in C_i$, and $T_2[y]=0$ otherwise (if $y$ is redundant).
	Using these two arrays, one can access in constant time to the number of vertices that are yet to be dominated in the unique clique $C_i$ in which $y$ belongs, by checking $T_1[T_2[y]]$.
	Let $M_1$ be a two dimensional array of size $n\times 2$ such that $M_1[a][0]=|Priv_{IR}(A,a)\cap C^a|$ if $C^a\neq \emptyset$, $M_1[a][0]=-1$ otherwise, and $M_1[a][1]=|Priv_{IR}(A,a)\setminus C^a|$.
	Let $M_2$ be an array of size $n$ such that $M_2[y]=a$ if $y\in Priv_{IR}(A,a)$, and $M_2[y]=0$ otherwise.
	Let $M_3$ be an array of size $n$ such that $M_3[y]=0$ if furthermore $y\in C^a$, $M_3[y]=1$ otherwise.
	Using these three arrays, one can access in constant time to the number of irredundant private neighbors a vertex $a$ such that $y\in Priv_{IR}(A,a)$ has by checking $M_1[M_2[y]][0]$ and $M_1[M_2[y]][1]$. 
	The size of the set $Priv_{IR}(A,a)\cap C^a$ in case where $y\in C^a$, and $Priv_{IR}(A,a)\setminus C^a$ in case where $y\not\in C^a$ can be accessed by $M_1[M_2[y]][M_3[y]]$.
	Finally, consider an array $W$ of size $n$ initialized to zero. 
	This array will be used to know if a vertex $y$ is dominated by $A\cup \{c\}$, by setting $W[y]=x$ if $y$ is connected to some $x\in A\cup \{c\}$ and $W[y]=0$ otherwise.
	Note that these six arrays can be computed in $O(n^2)$ preprocessing time and $O(n^2)$ space.

	We are now ready to detail each iteration of the loop Line~\ref{line:RN-forall}. 
	When considering a new candidate vertex $c\in RN(G)\setminus A$, we do the following.
	For each $y\in N(c)\cap IR(G)$, we set $W[y]:=c$, $M_2[y]:=c$ and $T_1[T_2[y]]:=T_1[T_2[y]]-1$ whenever $W[y]=0$ (i.e., if $y$ is not dominated by $A$).
	Note that $T_1[T_2[y]]$ is decreased to zero if and only if $c$ verifies $C_i\not\subseteq N(A)$ and $C_i\subseteq N(A\cup\{c\})$ for $i=T_2[y]$.
	The next claim follows.

	\begin{claim}\label{claim:P7-1}
		Deciding whether there exists $C_i$, $i\in \intv{1}{\ell}$ such that $C_i\not\subseteq N(A)$ and $C_i\subseteq N(A\cup\{c\})$ takes $O(\deg(c))$ time.
	\end{claim}

	If~$W[y]\neq 0$ then $y$ was already dominated by $A$, and in particular it might have been the private neighbor of some $a\in A$ given by both $M_2[y]$ and $W[y]$.
	In that case (i.e., whenever $M_2[y]\neq 0$) we set $M_1[M_2[y]][M_3[y]]:=M_1[M_2[y]][M_3[y]]-1$, and $M_2[y]:=n+1$ (this value is set temporarily).
	Note that $M_1[M_2[y]][0]$ (resp.~$M_1[M_2[y]][1]$) is decreased to zero if and only if $c$ steals all the private neighbors of $a\in A$ that are in irredundant components that are partially adjacent (resp.~complete) to $a$. 
	Also, observe that we still have $W[y]=a$ for all such $y$ in that case.
	We prove the following

	\begin{claim}\label{claim:P7-2}
		Deciding whether every $a\in A$ has a private neighbor in some irredundant component $C_j\subseteq N(A\cup \{c\})$, $j\in \intv{1}{\ell}$ takes $O(\deg(c))$ time.
	\end{claim}

	\begin{claimproof}
		Consider some $y\in N(c)\cap IR(G)$ and let $a=M_2[y]$, $j=M_3[y]$.
		Observe that if both $M_1[a][0]$ and $M_1[a][1]$ have value zero after updating $M_1[a][j]:=M_1[a][j]-1$, then we answer negatively ($a$ lost all its private neighbors).
		If $M_1[a][1]$ does not equal zero, then we answer positively ($a$ has a private neighbor in a dominated irredundant component).
		If $M_1[a][1]$ equals zero and $M_1[a][0]$ does not equal zero, then we need to check whether $C^a$ is dominated or not, that is whether $T_1[T_2[y]]$ equals zero or not.
		We answer positively if it is the case, and negatively otherwise.
		This covers all possibilities and the claim follows.
	\end{claimproof}
	
	A consequence of Claims~\ref{claim:P7-1} and~\ref{claim:P7-2} is that $A\cup \{c\}\in \D_{RN}(G)$ can be decided in $O(\deg(c))$ time in the condition of Line~\ref{line:RN-condition}.
	Now, if $A\cup \{c\}\in \D_{RN}(G)$ then we set $M_2[y]=0$ whenever $M_2[y]=n+1$ ($y$ is adjacent to some $a\in A$ and $c$, it is not a private neighbor anymore), for every $y\in N(c)\cap IR(G)$ and in a time which is also bounded by $O(\deg(c))$.
	Let us overview the case where $c$ does not satisfy the conditions of Claims~\ref{claim:P7-1} and~\ref{claim:P7-2}, or when a backtrack is executed.
	First, we undo the changes by setting $W[y]=0$ and $T_1[T_2[y]]:=T_1[T_2[y]]+1$ for every $y\in N(c)\cap IR(G)$ such that $W[y]=c$ (such a $y$ was not adjacent to any $a\in A$ and no other modifications occurred).
	If $W[y]\neq c$ and $M_2[y]= n+1$ (in that case $y$ was the private neighbor of some $a\in A$, $a=W[y]$) then we set $M_2[y]=W[y]$ and $M_1[M_2[y]][M_3[y]]=M_1[M_2[y]][M_3[y]]+1$.
	If $W[y]\neq c$ and $M_2[y]\neq n+1$ then $y$ was not adjacent to $c$ and no modification occurred.
	This undo process also takes $O(\deg(c))$ time.
	Since the sum of degrees of $G$ is bounded by $O(n+m)$, the time spent in the loop Line~\ref{line:RN-forall} is bounded by $O(n+m)$.

	Let us finally consider the case when consecutive backtracks are executed.
	Observe that in that case, it could be that $n$ times $O(n+m)$ steps are computed without output.
	In order to avoid this, a common trick is to output half of the solutions while going down the recursive tree, and the other half when going up the tree.
	This is done by moving the output of Line~\ref{line:RN-output} after the loop Line~\ref{line:RN-afterloop} on odd depths of the recursive tree.
\end{proof}

\subsection{Polynomial delay implementation in $P_8$-free chordal graphs.}\label{subsec:RN-P8}

We show that \textsc{IEP} and \textsc{MGP} can be solved in polynomial time in $P_8$-free chordal graphs.
This yields a polynomial-delay implementation of Algorithm~\ref{algo:RN} in the same class.

From now on and until the end of the section, let $G$ be a $P_8$-free chordal graph.
Recall that by Proposition~\ref{prop:Pk-Pk-4}, every irredundant component $C$ of $G$ induces a graph $H=G[C]$ that is $P_4$-free chordal.
It is known that every $P_4$-free chordal graph is the comparability graph of a tree poset~\cite{wolk1962comparability}, where two vertices of the graph are made adjacent if they are comparable in the poset.
To $H$ we associate $T(H)$ its tree poset.
Note that in particular, the root of $T(H)$ is universal in $H$, and that $x\leq y$ implies $N_H[y]\subseteq N_H[x]$.
An example of a $P_4$-free chordal graph and its tree poset is given in Figure~\ref{fig:P4}.

\begin{figure}
	\center
	\caption{A $P_4$-free chordal graph $H$ and the Hasse diagram of its poset tree.
	On such instance $p=4$, $X_1=\{t_1,1,4\}$, $X_2=\{5, t_3,t_4\}$, $X_3=\{t_6\}$, $X_4=\{t_7\}$ and $Y=\{3\}$.
	Then $F=\{t_2,2,t_5\}$ and a set $D$ such that $D$ dominates $C\setminus (X\cup Y)$ but not $X_1,\dots,X_4$ is given by $D=\{t_2,t_3,t_5\}$.}\label{fig:P4}
	\includegraphics{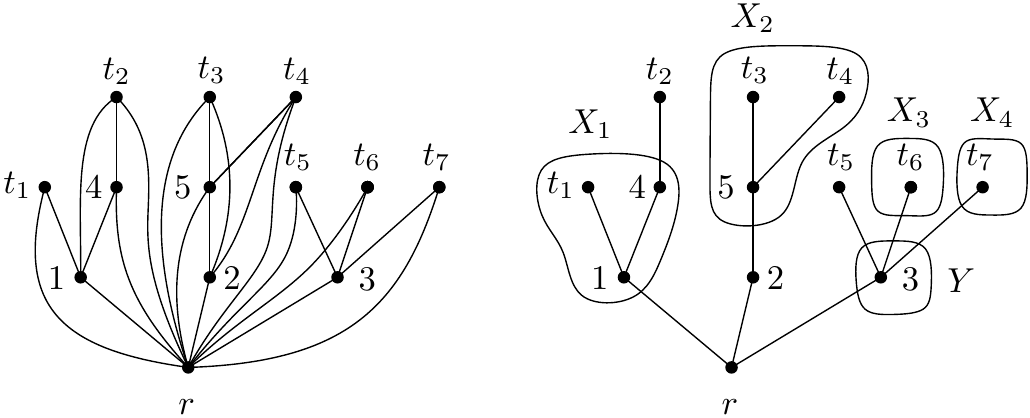}
\end{figure}

In the following, let $A\subseteq RN(G)$ and $C$ be an irredundant component of $G$. 
Let $x_1,\dots,x_p$ denote the $p$ elements of $R(A)$ having a private neighbor in $C$, $X_j=Priv_{IR}(A, x_j)$ for every $j\in \intv{1}{p}$, $X=X_1\cup\dots\cup X_p$ and $Y=(N(A)\cap C)\setminus X$ (this last set correponds to the vertices in $C$ that are already dominated by $A$, but that are not private for any $x_i$, $i\in \intv{1}{p}$).
By Corollary~\ref{cor:P9-IR-characterization}, \textsc{IEP} can be tested independently on every such component by checking whether $X_1,\dots,X_p$ are non-empty, and, whether there exists $D\subseteq C$ such that
\begin{itemize}
	\item $D$ dominates $C\setminus (X\cup Y)$, but
	\item $D$ does not dominate any of $X_1,\dots,X_p$.
\end{itemize}
We will show that such a test can be conducted in linear time whenever $X_1,\dots,X_p$ and $Y$ are given by lists and arrays, a condition that can be fulfilled at low cost as in the implementation of Theorem~\ref{thm:P7}.
In the remaining of this section, we note $r$ the root of $T(H)$, and $F$ the maximal elements of $T(H)$ which are neither in $X_1,\dots,X_p$ nor in $Y$ (hence no two elements of $F$ are comparable in $T(H)$).
One such instance is given in Figure~\ref{fig:P4}.

\begin{lemma}\label{lem:P4-F}
	Let $D$ be a subset of vertices of $H$.
	Then $D$ dominates $C\setminus (X\cup Y)$ if and only if it dominates $F$.
\end{lemma}

\begin{proof}
	The first implication trivially holds since $F$ is selected in $C\setminus (X\cup Y)$.
	Now since every vertex of $C\setminus (X\cup Y)$ belongs to a path in $T(H)$ from $r$ to some $x\in F$, $F$ dominates $C\setminus (X\cup Y)$.
	Consider any $x\in F$ and some dominating set $D$ of $F$.
	Since $D$ dominates $F$, either $x\in D$, or there exists $y\in D$ such that either $x<y$, or $x>y$.
	In all such cases, the unique path from $r$ to $x$ in $T(H)$ is dominated.
	Hence $D$ dominates $C\setminus (X\cup Y)$.
\end{proof}

\begin{lemma}\label{lem:IEP-characterization}
	There exists a set $D$ dominating $C\setminus (X\cup Y)$ and not $X_1,\dots,X_p$ if and only if for every $x\in F$ there exists a leaf $t$ of $T(H)$ such that $x\leq t$ and $X_i\not\subseteq N_H[t]$, $i\in \intv{1}{p}$.
\end{lemma}

\begin{proof}
	We show the first implication.
	Let $D$ be a dominating set of $C\setminus (X\cup Y)$ which does not dominate $X_1,\dots,X_p$.
	By Lemma~\ref{lem:P4-F}, every $x\in F$ is dominated by some $y\in D$.
	Consider some such $x$ and $y$.
	Then one of $x\leq y$ or $y<x$ holds. 
	Let $t$ be a leaf of $T(H)$ such that $y\leq t$ and $x\leq t$.
	Then $t$ does not dominate any of $X_1,\dots,X_p$ as $y$ does not, and $y\leq t$ hence $N_H[y]\supseteq N_H[t]$.
	Since $x\leq t$ the first implication follows.
	As for the other implication, is it a consequence of Lemma~\ref{lem:P4-F} observing that every leaf $t$ such that $x\leq t$ for some $x\in F$ dominates $x$.
\end{proof}

The next lemma shows that the characterization of Lemma~\ref{lem:IEP-characterization} can be checked for each element $x\in F$ independently.

\begin{lemma}\label{lem:IEP-independently}
	Consider $X_j$, $j\in \intv{1}{p}$.
	Then, either
	\begin{itemize}
		\item $X_j\subseteq \{y\in C \mid x<y\}$ for some unique $x\in F$, or
		\item $X_j\subseteq \{y\in C \mid y<x\ \text{for some}\ x\in F\}$ and \textsc{IEP} can be answered negatively, or
		\item $X_j\not\subseteq N_H[F]$ and it can be ignored when checking the characterization of Lemma~\ref{lem:IEP-characterization}.
	\end{itemize}
\end{lemma}

\begin{proof}
	Let $X_j$, $j\in \intv{1}{p}$.
	First note that if $X_j\subseteq \{y\in C \mid x<y\}$ then such a $x$ is unique or else $T(H)$ is not a tree.
	Let us assume that $X_j\subseteq \{y\in C \mid y<x\ \text{for some}\ x\in F\}$.
	Observe that these two cases are disjoint as otherwise there exist two elements of $F$ that are comparable, a contradiction.
	Recall that by Proposition~\ref{lem:P4-F} every dominating set $D$ of $C\setminus (X\cup Y)$ dominates $F$.
	Now if $D$ dominates $F$ then it dominates every $y\in C$ such that $y<x$ for some $x\in F$, and $X_j$ consequently.
	In that case \textsc{IEP} can be answered negatively.

	Let us now assume that $X_j$ is not of the first two cases.
	We first show by contradiction that there is no $x\in F$ such that $x<y$ for some $y\in X_j$. 
	Suppose that there exist two such $x$ and $y$. 
	Since $X_j$ is not of the first case there must be some $y'$, $y'\neq y$ such that $x\not<y'$.
	Consider $a\in R(A)$ such that $X_j=Priv_{IR}(A,a)$.
	Note that $x$ belongs to a shortest path from $y$ to $y'$ in $C$ (the common ancestor of $y$ and $y'$ in $T(H)$ must be smaller than $x$).
	By Proposition~\ref{prop:Na-connected}, $x$ belongs to $N(a)\cap C$.
	Hence it either belongs to $X_j$, or $Y$, contradicting the fact that $x\in F$.
	Consequently, and since $X_j$ is not of the second case, $X_j\not\subseteq N_H[F]$.
	Now, since the leaves of $T(H)$ selected in Lemma~\ref{lem:IEP-characterization} are of neighborhood included in that of $F$, $X_j$ can be ignored.
\end{proof}

\begin{lemma}\label{lem:IEP-algorithm}
	There is an $O(n+m)$ time algorithm solving \textsc{IEP} whenever $G$ is $P_8$-free chordal, where $n$ and $m$ respectively denote the number of vertices and edges in $G$, whenever 
	\begin{itemize}
		\item the leaves of $T(H=G[C])$,
		\item the predecessors and the successors of every $x\in T(H)$, and
		\item each of the sets $X_1,\dots,X_p$ and $Y$
	\end{itemize}
	are given by lists and arrays for every irredundant component $C$ of $G$.
\end{lemma}

\begin{proof}
	Let us first focus on an irredundant component $C$ of $G$, and $H=G[C]$ its induced subgraph.
	We want to decide whether $F$ can be dominated without dominating $X_1,\dots,X_p$.
	Note that by assumption, the leaves of $T(H)$, the predecessors and the successors of every $x\in T(H)$, and each of the sets $X_1,\dots,X_p$ and $Y$ can be iterated in a time which is bounded by their size.
	Furthermore, deciding whether a vertex belongs to one given set takes constant time.
	The same assumptions hold for the set $Z=C\setminus (X\cup Y)$ which can be computed in $O(n_H)$ time iterating on $X_1,\dots,X_p$ and $Y$. 

	The algorithm proceeds as follows.
	First it computes $F$ by checking for every $x\in Z$ whether it has a successor in $Z$.
	It computes the set $F^-=\{y\in C \mid y<x\ \text{for some}\ x\in F\}$ in a $n$-element array by adding predecessors of every $x\in F$ at a time.
	Since the sum of degrees of $H$ is bounded by $O(n_H+m_H)$, this takes $O(n_H+m_H)$ time.
	Then it tests for every set $X_1,\dots,X_p$ whether it is included in $F^-$ within the same time.
	At this stage if we find an inclusion then we can answer negatively according to the second item of Lemma~\ref{lem:IEP-independently}, and can consider $X_1,\dots,X_p$ to be of the first type in the following.

	For every set $X_j$, $j\in \intv{1}{p}$ we check whether it has two non-adjacent vertices.
	This is done in $O(n_H+m_H)$ time by testing for every vertex in $X_j$ whether it has a neighbor in $X_j$, recalling that every such $X_j$ is disjoint (we iterate through vertices and their neighborhood only once).
	If $X_j$ has no two non-adjacent vertices, then it is a path in $T(H)$ and we mark in a $n_H$-element array the indexes of leaves that are greater that its maximal element (each of these leaves dominates $X_j$).
	Similarly, computing the maximal element of every such $X_j$, and the indexes of leaves that are greater that their maximal element, can be done in $O(n_H+m_H)$ time.
	If a set $X_j$, $j\in \intv{1}{p}$ has two non-adjacent vertices then no leaf $t$ of $T(H)$ dominates $X_j$ and it can be ignored for the next step.
	We now proceed as follows according to Lemma~\ref{lem:IEP-independently}.
	We check independently for every $x\in F$ if it has a descendant leaf $t$ (to each $x\in F$ corresponds disjoint sets of such leaves) which was not indexed previously.
	If it has then we answer positively. 
	If not then $x$ (hence $F$) cannot be dominated without dominating one of $X_1,\dots,X_p$ and we can answer negatively.
	
	We now need to conduct this test for every irredundant component $C$ of $G$ independently.
	Since irredundant components are subgraphs of $G$ we have that $n$ and $m$ are respectively bounded by the sums of $n_H$'s and $m_H$'s for every irredundant component $C$ of $G$ where $H=G[C]$, and the complexity follows.
\end{proof}

A corollary of Lemma~\ref{lem:IEP-algorithm} is the following, observing that $x$ is a maximal generator of $A$ if and only if $A\setminus \{y\}\not\in \D_{RN}(G)$ for any $y\in A$ of index greater than $x$, and that $n$ times $O(n+m)$ is bounded by $O(n\cdot m)$ since $G$ is connected.

\begin{corollary}\label{cor:MGP-algorithm}
	There is an $O(n\cdot m)$ time algorithm solving \textsc{MGP} whenever $G$ is $P_8$-free chordal, where $n$ and $m$ respectively denote the number of vertices and edges in $G$, and assuming the conditions of Lemma~\ref{lem:IEP-algorithm}. 
\end{corollary}

We can thus conclude the section with the following result.

\begin{theorem}\label{thm:P8-RN}
	There is an $O(n^2\cdot m)$ delay and $O(n^2)$ space implementation of Algorithm~\ref{algo:RN} whenever $G$ is $P_8$-free chordal, where $n$ and $m$ respectively denote the number of vertices and edges in $G$.
\end{theorem}

\begin{proof}
	By~Lemma~\ref{lem:IEP-algorithm} and Corollary~\ref{cor:MGP-algorithm}, there is an $O(n^2\cdot m)$ delay implementation of Algorithm~\ref{algo:RN} whenever the assumptions of Lemma~\ref{lem:IEP-algorithm} can be fulfilled at every step of the loop Line~\ref{line:RN-forall}.
	Clearly, the representation $T(H)$ of $H=G[C]$ can be computed for every irredundant component $C$ of $G$ in $O(n^2)$ preprocessing time, and $O(n^2)$ space.
	The lists and arrays containing the leaves of $T(H)$, the predecessors and the successors of every $x\in T(H)$, and that will contain the sets $X_1,\dots,X_p$, $Y$ and $Z=C\setminus (X\cup Y)$ at each step of loop Line~\ref{line:RN-forall} can also be computed within these preprocessing-time and space complexities.
	Furthermore, the sets $X_1,\dots,X_p$ and $Y$ can be maintained at each step of loop as in the proof of Theorem~\ref{thm:P7-RN}, and in a time which is clearly upper-bounded by $O(n\cdot m)$ for each $c\in RN(G)\setminus A$.
	We proceed as in the proof of Theorem~\ref{thm:P7-RN} to maintain an $O(n^2\cdot m)$ delay in case of consecutive backtrack.
	The theorem follows.
\end{proof}

\section{Enumerating irredundant extensions}\label{sec:irredundant}

This section is devoted to the complexity analysis of Line~\ref{line:main-forallIR} of Algorithm~\ref{algo:main}.
More precisely, we show that irredundant extensions can be enumerated with linear and polynomial delays in $P_7$-free and $P_8$-free chordal graphs.
This allows us to conclude with the two main results of this paper.

\subsection{Irredundant extensions in $P_7$-free chordal graphs.}

Let $G$ be a $P_7$-free chordal graph and $A\in \D_{RN}(G)$.
Recall that by Corollary~\ref{cor:P7-IR-characterization}, every $x\in A$ has a private neighbor in some irredundant component $C\subseteq N(A)$.
Let $C_1,\dots,C_k$ denote the $k$ irredundant components of $G$ that are not dominated by $A$.
By Proposition~\ref{prop:Pk-Pk-4}, every such component is a clique.
Consequently we have that
\[
	\DIR(A)=	\{\{x_1,\dots,x_k\} \mid x_i\in C_i,\ i\in\intv{1}{k}\}.
\]
Now, such a set can clearly be enumerated with $O(n+m)$ delay given $A$ and $C_1,\dots,C_k$.
Furthermore, a track of these irredundant components is maintained at each step of the loop Line~\ref{line:RN-forall} of Algorithm~\ref{algo:RN} in the implementation of Theorem~\ref{thm:P7-RN}.
We conclude with the following theorem which improves a previous result of Kant\'e et al.~in~\cite{kante2014enumeration} on $P_6$-free chordal graphs.

\begin{theorem}\label{thm:P7}
	There is an $O(n+m)$ delay, $O(n^2)$ space and $O(n^2)$ preprocessing-time algorithm enumerating $\D(G)$ whenever $G$ is $P_7$-free chordal, where $n$ and $m$ respectively denote the number of vertices and edges in $G$.
\end{theorem}

\subsection{Irredundant extensions in $P_8$-free chordal graphs.}

Let $G$ be a $P_8$-free chordal graph and $A\in \D_{RN}(G)$.
By Theorem~\ref{thm:P9-IR-characterization}, the intersection of an irredundant extensions of $A$ with an irredundant component $C$ of $G$ is a minimal set $D\subseteq C$ such that
\begin{itemize}
	\item $D$ dominates $C\setminus (X\cup Y)$, and
	\item $D$ does not dominate any of $X_1,\dots,X_p$,
\end{itemize}
where $x_1,\dots,x_p$ denote the $p$ elements of $R(A)$ having a private neighbor in $C$, where $X_j=Priv_{IR}(A, x_j)$ for every $j\in \intv{1}{p}$, $X=X_1\cup\dots\cup X_p$ and $Y=(N(A)\cap C)\setminus X$.
In the following, to $A$ and $C$ we associate $\DIR(A,C)$ the set of all such minimal sets.
Then, if $C_1,\dots,C_k$ denote the $k$ irredundant components of $G$ that are not dominated by $A$, we have that
\[
	\DIR(A)= \{D_1\cup \dots \cup D_k \mid D_i\in \DIR(A,C_i),\ i\in\intv{1}{k}\}.
\]
Clearly, such a set can be enumerated with $O(n^3\cdot m)$ delay given an algorithm enumerating $\DIR(A,C_i)$ with $O(n^2\cdot m)$ delay for every irredundant component $C_1,\dots,C_k$, where $n$ and $m$ respectively denote the number of vertices and edges in $G$. 
We shall show that such an algorithm exists.

Consider $C$, $X_1,\dots,X_p$ and $Y$ as described above.
Let $H=G[C]$.
Recall that by Proposition~\ref{prop:Pk-Pk-4}, $H$ is $P_4$-free chordal.
In the remaining of this section, we rely on the notations of Section~\ref{sec:redundant} and note $r$ the root of $T(H)$ and $F$ the maximal elements of $T(H)$ which are neither in $X_1,\dots,X_p$ nor in $Y$.
One such instance is given in Figure~\ref{fig:P4}.
We call {\em irredundant component extension problem}, denoted by \textsc{ICEP}, the following decision problem.
Given $S,Q\subseteq C$, is there a solution $D\in \DIR(A,C)$ such that $S\subseteq D$ and $D\cap Q=\emptyset$?
We shall show that this problem can be solved in $O(n\cdot m)$ time, which, using the {\em backtrack search} technique, leads to an $O(n^2\cdot m)$ algorithm enumerating irredundant extensions in $P_8$-free chordal graphs.

\begin{lemma}\label{lem:ICEP}
	There is an algorithm solving \textsc{ICEP} in $O(n\cdot m)$ time, assuming the conditions of Lemma~\ref{lem:IEP-algorithm}. 
\end{lemma}

\begin{proof}
	Observe that \textsc{IEP} restricted to a single component and \textsc{ICEP} only differ on the fact that the set $D\subseteq C$ should in addition satisfy $D\cap Q=\emptyset$ and should not dominate $Priv_H(S,s)\setminus (X\cup Y)$ for any $s\in S$, where $Priv_H(S,s)$ denotes the private neighborhood of $s\in S$ in $H$.
	In that case, $D$ can be reduced to a minimal set $D^*$ such that $S\subseteq D^*$ and $D^*\cap Q=\emptyset$.
	We show that these additional conditions can be handled at the cost of an increasing complexity, relying on the proof of Lemma~\ref{lem:IEP-algorithm}.
	Clearly, we can first answer negatively if $Priv_H(S,s)\setminus (X\cup Y)$ is empty for some $s\in S$. 
	Otherwise, the condition that $D$ does not dominate $Priv_H(S,s)\setminus (X\cup Y)$ for any $s\in S$ can be handled by adding extra sets $X_{p+i}=Priv_H(S,s_i)\setminus (X\cup Y)$ for every $s_1,\dots,s_q\in S$, and updating $X:=X\cup X_{p+1}\cup \dots \cup X_{p+q}$.
	Since these sets are connected Lemma~\ref{lem:IEP-independently} still applies.
	As for $D$ satisfying $D\cap Q=\emptyset$, we proceed as follows.
	For every $x\in F$, and instead of only checking descendant leaves of $x$, we iterate through all the descendants of $x$ and check whether it has a successor $y$ such that $y\not\in Q$, and such that $y$ does not dominate any of $X_1,\dots,X_{p+q}$.
	This can be done in $O(n\cdot m)$ time as we iterate through every such $y$ in $O(n+m)$ time, and check for each of these $y$'s whether it has some $X_j$ in its neighborhood in $O(n)$ time.
	At this stage, and according to Lemmas~\ref{lem:IEP-characterization},~\ref{lem:IEP-independently} and~\ref{lem:IEP-algorithm}, we can answer yes if and only every $x$ has such a neighbor $y$.
\end{proof}

We can now conclude using the backtrack search technique that we briefly recall now.
Formal proofs are omitted.
The enumeration is a depth-first search of a tree whose nodes are partial solutions and leaves are solutions. 
The algorithm constructs partial solutions by considering one vertex $x_i$ at a time (following some linear ordering $x_1,\dots,x_n$ of the vertices), checking at each step whether there is a final solution $D_1\in \DIR(A)$ containing $S\cup \{x_i\}$ and not intersecting $Q$, and one $D_2\in \DIR(A)$ containing $S$ and not intersecting $Q\cup \{x_i\}$.
This step is called the extension problem.
The algorithm recursively calls on such sets each time the extension is possible.
At first, $S$ and $Q$ are empty.
The delay time complexity is bounded by the depth of the tree (the number of vertices) times the time complexity of solving the extension problem, i.e., \textsc{ICEP}.
For further details on this technique, see for instance~\cite{read1975bounds,strozecki2019efficient}.

We conclude to the next lemma and theorem, noting that the conditions of Lemma~\ref{lem:IEP-algorithm} can be fulfilled as in the proof of Theorems~\ref{thm:P7-RN} and~\ref{thm:P8-RN}.

\begin{lemma}\label{lem:P8-IR}
	There is an algorithm enumerating $\DIR(A,C)$ with $O(n^2\cdot m)$ delay, assuming the conditions of Lemma~\ref{lem:IEP-algorithm}. 
\end{lemma}

\begin{theorem}\label{thm:P8}
	There is an $O(n^3\cdot m)$ delay and $O(n^2)$ space algorithm enumerating $\D(G)$ whenever $G$ is $P_8$-free chordal, where $n$ and $m$ denote the number of vertices and edges in $G$.
\end{theorem}

\section{Discussions}\label{sec:conclusion}

We investigated the enumeration of minimal dominating sets from their intersection with redundant vertices.
This technique was first introduced in~\cite{kante2014enumeration} and led to linear-delay algorithms in split and $P_6$-free chordal graphs.
We investigated generalizations of this technique to $P_k$-free chordal graphs for larger integers $k$.
In particular, we gave $O(n+m)$ and $O(n^3\cdot m)$ delays algorithms in the classes of $P_7$-free and $P_8$-free chordal graphs, where $n$ and $m$ respectively denote the number of vertices and edges in the graph.

As for $P_k$-free chordal graphs for $k\geq 9$, we now give evidence that the enumeration of $\D_{RN}(G)$ might need other techniques for $k\geq 9$, as \textsc{IEP} becomes {\sf NP}-complete.

\begin{theorem}\label{thm:NPC}
	\textsc{IEP} is {\sf NP}-complete even when restricted to $P_9$-free chordal graphs.
\end{theorem}

\begin{proof}
	First notice that \textsc{IEP} belongs to {\sf NP}: a polynomial certificate is given by an irredundant set $I\subseteq IR(G)$ such that $A\cup I\in \D(G)$, and which can be verified in polynomial time.

	Given an instance $\varphi$ of \textsc{3SAT} with variables $x_1,\dots,x_n$ and clauses $C_1,\dots,C_m$, we construct a $P_9$-free chordal graph $G$ and a set $A\subseteq RN(G)$ such that $A$ admits an irredundant extension if and only if there exists a truth assignment of the variables of $\varphi$ that satisfies all the clauses.
	In the following, we assume that the degenerate cases where a literal intersects every clause, where two clauses are equal, or where the number of variables and clauses is lesser than three are excluded.
	Then, the construction is the following.

	The first part concerns the construction of a split graph $H$ which contains one vertex for each of the literals $x_i$ and $\neg x_i$, a copy $u$ and $\neg u_i$ of such literals, and one vertex $c_j$ per clause $C_j$.
	The graph induced by the $u_i$'s, $\neg u_i$'s and $c_j$'s is completed into a clique, while an edge is added between $u_i$ and $x_i$, between $\neg u_i$ and $\neg x_i$, and between a literal $x_i$ (resp.~$\neg x_i$) and a clause $c_j$ whenever the literal is contained into that clause.
	As for the second part, it consists of a pendant path $x_iy_iz_i$ and $\neg x_i\neg y_i\neg z_i$ rooted at every literal $x_i$ and $\neg x_i$, and of a paw $a_ib_iv_iw_i$ (a triangle $a_iv_iw_i$ with a pendant edge $a_ib_i$) made adjacent to both $u_i$ and $\neg u_i$ only through $v_i$, for every $i\in \intv{1}{n}$.
	The construction is illustrated in Figure~\ref{fig:NPC}.
	It can be easily seen that the obtained graph $G$ is $P_9$-free chordal.
	Also that a $P_8$ is induced by $b_ia_iv_iu_iu_jv_ja_jb_j$ for $i\neq j\in\intv{1}{n}$.

	\begin{figure}
		\center
		\caption{The construction of $G$ in Theorem~\ref{thm:NPC}. Irredundant vertices are represented in black while redundant ones are in white. The vertices in the rectangle induce a clique and $H$ a split graph.}
		\includegraphics{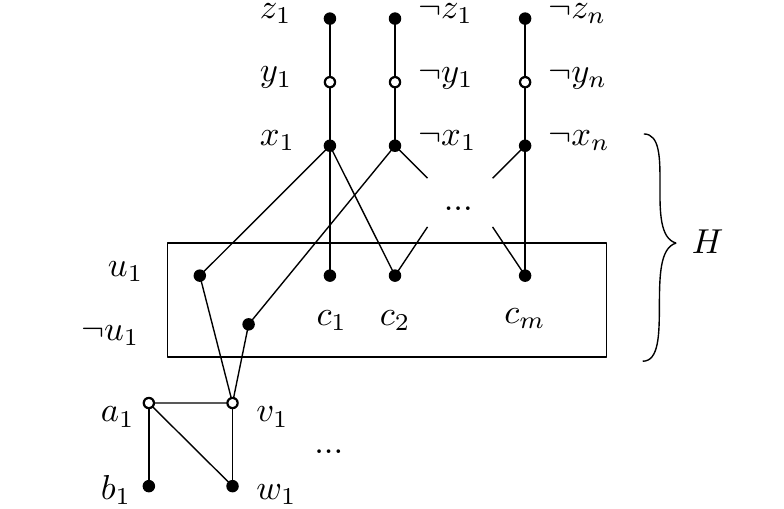}
		\label{fig:NPC}
	\end{figure}

	Let us show that $H=G[C]$ for an irredundant component $C$ of $G$.
	First note that every vertex outside of $H$ has a neighbor that is not adjacent to $H$, so it cannot make a vertex from $H$ redundant.
	Now if a vertex of $H$ makes another one of $H$ redundant then it cannot be a literal or some $u_i$, $\neg u_i$ (as it has either $y_i$, $\neg y_i$ or $v_i$ as a neighbor outside of $H$).
	Also, it cannot be a clause as by assumption every two clauses differ on a literal, and no literal is complete to the clique.
	Hence vertices of $H$ are all irredundant.
	It is easily seen that irredundant components of $G$ include $\{z_i\}$, $\{\neg z_i\}$, $\{b_i\}$ and $\{w_i\}$ for all $i\in\intv{1}{n}$.
	Also that redundant vertices of $G$ are $a_i$'s, $v_i$'s, $y_i$'s and $\neg y_i$'s.
	We conclude that $H$ cannot be extended, hence that $C$ is indeed an irredundant component of $G$, as claimed.
	
	Now, let us set $A=RN(G)$ and show that $A\in \D_{RN}(G)$ if and only if there exists a truth assignment of the variables of $\varphi$ that satisfies all the clauses.
	If $A\in \D_{RN}(G)$ then there exists an irredundant extension $D\in \DIR(A)$.
	Observe that only the $c_i$'s are to be dominated by $D$, i.e., $IR(G)\setminus N(A)=\{c_1,\dots,c_m\}$.
	However, $D$ does not intersect any element of the clique of $H$ as otherwise it would dominate $\{u_i,\neg u_i\}$ and thus steal all the private neighbors of the $v_i$'s.
	For the same reason, it cannot contain one literal and its negation.
	Consequently $D$ corresponds to a truth assignment of the variables of $\varphi$ that satisfies all the clauses.
	Consider now any truth assignment of the variables of $\varphi$ that satisfies all the clauses, and $D$ the associated set of vertices in $G$.
	By construction $D$ dominates all the $c_i$'s.
	Furthermore it does not steal any private neighbor to any vertex of $A$.
	By Corollary~\ref{cor:IR-private} we have that $A\in \D_{RN}(G)$ concluding the proof.
\end{proof}

\bibliography{main}

\end{document}